\documentclass[aps,amssymb,amsmath, twocolumn, pra, superscriptaddress]{revtex4}
\usepackage{graphicx}
\usepackage{epsfig}
\usepackage{dcolumn}
\usepackage{bm}
\usepackage{bbm}
\usepackage{soul}
\usepackage[colorlinks, citecolor=blue]{hyperref}
\usepackage{color}
\usepackage[up]{subfigure}
\usepackage{ulem} 
\usepackage{soul}
\usepackage{amssymb}
\usepackage{amsmath}
\newcommand{\be}{\begin{equation}}
\newcommand{\ee}{\end{equation}}
\newcommand{\bc}{\begin{center}}
\newcommand{\ec}{\end{center}}
\newcommand{\bea}{\begin{eqnarray}}
\newcommand{\eea}{\end{eqnarray}}
\newcommand{\ba}{\begin{array}}
\newcommand{\ea}{\end{array}}

\newtheorem{theo}{Theorem }

\newenvironment{proof}[1][Proof  ]{\begin{trivlist}
\item[\hskip \labelsep {\bfseries #1}]}{\end{trivlist}}
\newcommand{\qed}{\nobreak \ifvmode \relax \else
      \ifdim\lastskip<1.5em \hskip-\lastskip
      \hskip1.5em plus0em minus0.5em \fi \nobreak
      \vrule height0.75em width0.5em depth0.25em\fi}

\begin{document}
\title{\large LOCALIZED QUANTUM WALKS AS SECURED QUANTUM MEMORY}

\author{C. M. Chandrashekar \footnote{Current Address : The Institute of Mathematical Sciences,\\ CIT Campus - Taramani, Chennai, India}}
\email{chandru@imsc.res.in}
\author{Th.~Busch}
\affiliation{Quantum Systems Unit, Okinawa Institute of Science and Technology Graduate University, Okinawa 904-0495, Japan}



\begin{abstract}
We show that a quantum walk process can be used to construct and secure quantum memory. More precisely, we show that a localized quantum walk with temporal disorder can be engineered to store the information of a single, unknown qubit on a compact position space and faithfully recover it on demand. Since the localization occurss with a finite spread in position space, the stored information of the qubit will be naturally secured from the simple eavesdropper. Our protocol can be adopted to any quantum system for which experimental control over quantum walk dynamics can be achieved.
\end{abstract}


\maketitle


\section{Introduction}
\label{intro}
Quantum memories are strategies that allow to store and retrieve the state of an unknown quantum bit faithfully.  While the physics describing the necessary conversion processes is often interesting from a fundamental point of view, the processes themselves are of importance for applications in quantum information and communication. Various strategies have been considered to store and retrieve quantum states of light\,\cite{MHN97, CMJ05, Eis05, AFK08, CDL08, LKF09, HLL10} and impressive progress has been made in the field of ensemble-based quantum memories\,\cite{Ima09, WAB09, MWT10, YYH11, LST09, JPB12, NJP13}. However, the quest for simple, long-lived, secured and system independent  protocols for quantum memory continues and here we suggest that a discrete-time quantum walk with temporal disorder can be engineered to compactly store single qubits in position space and recover them on demand. Since the localization occurs with a finite spread in position space, the stored information will be naturally secured from an eavesdropper who does not have access to the complete position space. The storage time will be directly related to the number of implementable steps of the walk.

Discrete-time quantum walks\,\cite{Ria58+,ADZ93}, described by a quantum coin operation followed by a shift operation, evolve a localized quantum state into a coherent superposition of different locations in position space.  They are known to have algorithmic applications, which allow to solve a number of problems more efficiently than purely classical approaches\,\cite{Amb03}. Furthermore, it has been shown that they can be used to implement universal quantum computation\,\cite{Chi09, LCE10} or construct a generalized measurement device\,\cite{KW13}. Recently astonishing experimental progress in controlling the dynamics of single quantum states has led to implementations of quantum walks in the NMR system, ions, photons, and atoms\,\cite{DLX03+, K09}. Investigating the possibilities of using quantum walks as constituent in quantum computers and for other fundamentally important quantum information processing and communication applications is therefore a promising challenge.  One such use we present here is a quantum memory which is a must for any quantum computing device\,\cite{PK11}.

At a first look a dynamical process like a quantum walk appears to be an unlikely candidate for a quantum memory, since it results in nontrivial quantum correlations between the particle (qubit) and the position space. While this alters the state of the qubit as a function of time, we present in the following a careful analysis of the dynamics and show that the stored information can be perfectly recovered  at specific  times $t$. These times are periodic and a function of the coin parameter $\theta$, used for evolving the walk. Due to the spatial spread of the qubit in position space, the information stored will also acquire an inherent level of security from an eavesdropper.

Though in principle this model describes a fully functioning quantum memory, the dependence of the recovery time $t$ on $\theta$ poses an unwelcome and limiting restriction. Additionally, the size of the position space required to store the information of the qubit increases linearly over time, thus adding an experimental challenge for achieving long storage time. Surprising as it seems, both these issues can be overcome.

To enable the recovery of the information at any time and independently of the choice of $\theta$, we show that the presented protocol can be amended by using a Hadamard operation to encode and decode the initial state of the qubit before and after the walk, respectively. To limit the required size in position space  we propose to use one of the remarkable effects of quantum mechanics, that is, localization of a quantum state (particle) in presence of disorder\,\cite{And58, LR85}. 
This phenomenon, commonly known as Anderson localization, is usually discussed in terms of coherent evolution (e.g. a quantum walk) in the presence of a disordered medium. However, in a quantum walk disorder can be introduced in ways other than through the medium alone\,\cite{BCA03, SBB03, RAS05, OKA05, JM10, YKE08, Kon10+, AVW11, Joy11, Cha11a, Cha12+}. By breaking the periodicity of the evolution through randomizing the operations which determine the dynamics of the system, the effect of a random medium can be mimicked and localization around the origin achieved\,\cite{JM10, Cha11a, Cha12+}. Such disordering operations can be designed to cause spatial disorder, temporal disorder or spatio-temporal disorder by choosing appropriate quantum coin operations at different positions in space, time, or both, respectively. Among the three forms of disorder-creating evolutions, symmetry  is known to be preserved by the temporal disorder of the particular form presented in Ref.\,\cite{Cha12+} and we show that the corresponding localization ensures that the information of the qubit can be stored compactly in position space. 

\begin{figure}[tb]
\bc 
\includegraphics[width=8.3cm]{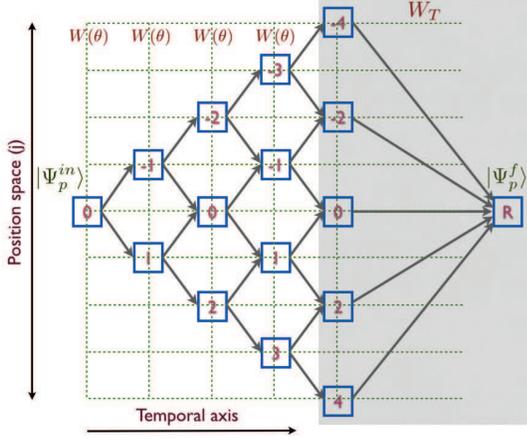} 
\ec
\vskip -0.3cm
\caption{\footnotesize{{\bf Schematic for the discrete-time quantum walk as quantum memory.} The qubit, initially in state $|\Psi_p^{in}\rangle$ at position $j=0$, undergoes quantum walk evolution for a time $t$ before it is retrieved.  The spread in position space during the evolution allows to secure the information of the qubit from an eavesdropper. \label{fig:1}}}
\end{figure}

\section{Discrete-time quantum walks as quantum memory} 
\label{QWL}

We consider a two-state particle initially in the state $|\Psi^{in}_p\rangle$ and located at the origin.  Using the basis states $|0\rangle=\begin{bmatrix} 1  \\ 0 \end{bmatrix}$ and $| 1 \rangle=\begin{bmatrix} 0  \\ 1 \end{bmatrix}$ we can write $|\Psi_{in}\rangle= \Big [ \cos(\delta)| 0 \rangle + e^{i\eta}\sin(\delta)|1 \rangle \Big ]\otimes | j=0\rangle$, where $j \in {\mathbbm  I}$ labels the position Hilbert space. Each step of the quantum walk consists of a coin operation 
\bea
\label{coinOP}
     B (\theta)\equiv  
     \begin{bmatrix}
     \mbox{~~~}\cos(\theta)      &    -i \sin(\theta)\\
      -i\sin(\theta)  &  \mbox{~~~}\cos(\theta) 
     \end{bmatrix},
\eea
which evolves the particle into a superposition of its basis states, followed by the shift operation $S \equiv     \sum_j \Big [  |0 \rangle\langle
0|\otimes|j-1 \rangle\langle   j|   +  | 1 \rangle\langle
1|\otimes |j+1 \rangle\langle j| \Big ]$.
This translates the internal state into a superposition in position space
\be
 \label{Wop}
 W(\theta)\equiv S  \Big [ B (\theta) \otimes  {\mathbbm 1} \Big ].
\ee
The state after time $t$, with unit time for each step, is then given by
\bea
   |\Psi_t\rangle=[W (\theta)]^t|\Psi_{\it in}\rangle = \sum_j \left[ \alpha_{j,t} |0 \rangle + \beta_{j,t} |1\rangle\right]\otimes|j\rangle, 
\eea
where
\begin{align}
  \label{eq:comp}
  \alpha_{j,t} &= \cos(\theta)\alpha_{j+1,t-1} -i\sin(\theta)\beta_{j+1, t-1} \\
  \label{eq:5}
    \beta_{j,t} &= \cos(\theta)\beta_{j-1,t-1} - i \sin(\theta)\alpha_{j-1,t-1}.
\end{align}
To retrieve the quantum state at the end of the walk (i.e.~to read out the memory, all parts have to be collected to a single vertex in position space, which we label $R$ (see Fig.~\ref{fig:1}). As a physical process this can be done by connecting all nodes $j$ occupied at time $t$ to the final vertex $R$ and let all parts of the wavefunction interfere. Mathematically this collection process, $W_T$, can be described using the collection operators $C^0_{j,R}$ and $C^1_{j,R}$, which shift the states $|0\rangle$ and $|1\rangle$, respectively.  Its explicit form is given by 
\bea
\label{transfer}
W_{T} = \sum_{j^{\prime} =-t}^{t, R}\sum_{j =-t}^{t, R} C^1_{j^{\prime},R} C^0_{j, R},
\eea
where
\begin{align}
C^0_{j,R} & = |0\rangle \langle 0| \otimes |R\rangle \langle j | + |1\rangle \langle 1| \otimes |j\rangle \langle j | ,  \\ 
C^1_{j^{\prime}, R} &=  |0\rangle \langle 0| \otimes |j^{\prime}\rangle \langle j^{\prime} | + |1\rangle \langle 1| \otimes |R \rangle \langle j^{\prime} |.
 \end{align}
These operators, $C^{0}_{j, R}$ and $C^{1}_{j^{\prime}, R}$ are of the same form as the shift operators recently suggested for a directed discrete-time quantum walk\,\cite{HM09}, and also correspond to the operators used for translation of basis state in the experimental realization of a topological quantum walk\,\cite{KBF12} and for spin-dependent transport of atoms in an optical lattice\,\cite{MGW03}. 
The operation $W_T$ therefore ensures a shift of both basis states, $|0\rangle$ and $|1\rangle$,  from all spatial positions $j$ to $R$, which looses the information of the position from which the individual parts of the wavefunction come from and results in interference at $R$.  While in general interference during the quantum walk and at position $R$ can prevent the reconstruction of the initially stored information of the qubit, we present below a careful analysis to determine the conditions under which this interference is constructive. 

The final state retrieved after time $t$ at vertex $R$ is therefore given by
\bea
|\Psi^{f}_p(t)\rangle \otimes |R\rangle = W_{T} [W(\theta)]^t \Big [ |\Psi^{in}_p\rangle \otimes | j= 0 \rangle \Big ]
\label{trans}
\eea
where 
\bea
|\Psi^{f}_p(t)\rangle = \sum_{j=-t}^{t} \Big [ \alpha_{j, t} |0\rangle + \beta_{j,t}|1\rangle \Big ],
\eea
and if  $|\Psi^{f}_p(t)\rangle \equiv |\Psi^{in}_p\rangle$, the quantum walk can in principle be used as quantum memory.

\begin{figure}[tb]
\bc 
\includegraphics[width=8.6cm]{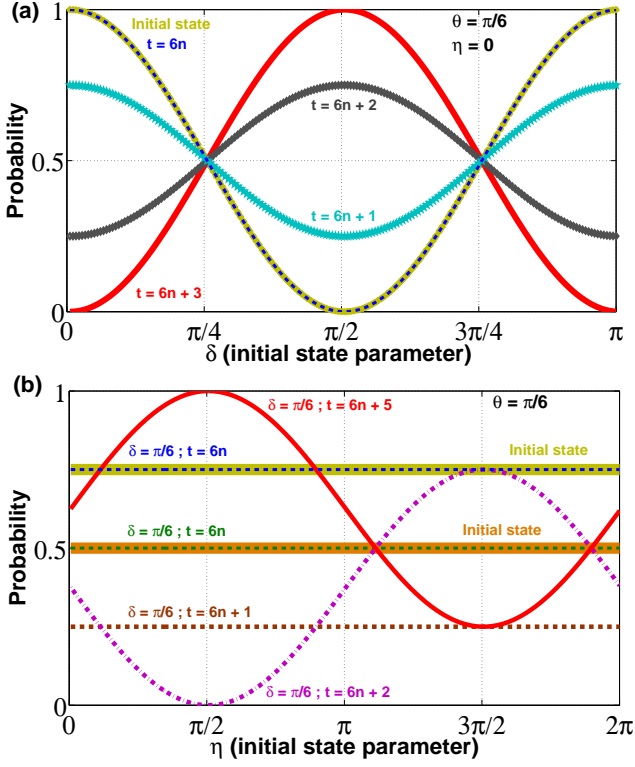} 
\ec
\caption{\footnotesize{
{\bf Probability of finding the basis state $|0\rangle$ in the state retrieved after a quantum walk with $t$ steps.} 
For the qubit with initial state $|\Psi_p^{in}\rangle = \cos(\delta) |0\rangle + e^{-i\eta}\sin(\delta)|1\rangle$ the probability of finding state $|0\rangle$ after retrieving the information at time $t$ (final state) is shown as function of  $\delta$ when $\eta =0$ in (a) and as function of $\eta$ for $\delta =\pi/6$ and $\pi/4$ in (b).  The value of $\theta=\pi/6$ for both (a) and (b). The initial probabilities ($t=0$) are indicated by the thick green and orange lines.  Both figures clearly show that the initial and final state are identical if either $t=n\pi/\theta$ or if $\delta =\pi/4$ or $3\pi/4$ for $\forall~ t$.
\label{fig:2}}}
\end{figure}

\begin{theo}
\label{Th1}
In a discrete-time quantum walk, using Eq.\,(\ref{coinOP}) as the coin operation, the initial state, $|\Psi^{in}_p\rangle= \alpha_0 | 0 \rangle + \beta_0 |1 \rangle$, is related to the final state retrieved after $t$ steps of walk as
\begin{align}
  |\Psi^f_p (t)\rangle   = e^{-i t \theta \cdot \sigma_x} |\Psi^{in}_p\rangle
                                          \end{align}
where $\sigma_x   =  \begin{bmatrix} \begin{array}{clcr}
  0      &     &   1
  \\ 1 & &    0 
\end{array} \end{bmatrix}$.
\end{theo}

\begin{proof} 
To show that Theorem\,\ref{Th1} is valid for a walk of any time $t+\tau$, we will use a backward iterative approach.  The state retrieved at the vertex $R$ after time $t$ will be the sum of the states spread across the position space,
\begin{align}
  |\Psi^f_p& (t)\rangle = \sum_{j= -t}^{t}\Big[ \alpha_{j, t} |0\rangle+\beta_{j, t} |1\rangle\Big].
  \end{align}
Since the state at position $j$ at time $t$ is dependent on the state at positions $j\pm1$ at time $t-1$ as given by Eq.\,(\ref{eq:comp}) and Eq.\,(\ref{eq:5}), the final state can be rewritten as
\begin{align}
\label{Eq1Th}
  |\Psi^f_p (t)\rangle & =
        \sum_{j= -t}^{t} \Bigg[\Big ( \cos(\theta) \alpha_{j+1, t-1} -i\sin(\theta) \beta_{j+1, t-1}\Big ) |0\rangle  \nonumber \\
                              &         + \Big ( \cos(\theta) \beta_{j-1, t-1} -i\sin(\theta) \alpha_{j-1, t-1} \Big ) |1\rangle \Bigg ].
\end{align}
Expanding this expression
\begin{align}
\label{eq:13}
  |\Psi^f_p  (t)\rangle  =   \Big[&\cos(\theta) \Big ( \alpha_{-t+1, t-1}+\alpha_{-t+2, t-1} + \cdots \alpha_{t, t-1}+\alpha_{t+1, t-1} \Big )\nonumber\\
 -i   \sin(\theta) \Big ( &\beta_{-t+1, t-1}+\beta_{-t+2, t-1} + \cdots \beta_{t, t-1}+\beta_{t+1, t-1} \Big ) \Big ] | 0\rangle \nonumber \\
         +     \Big[\cos(\theta) \Big ( &\beta_{-t-1, t-1}+\beta_{-t, t-1} + \cdots \beta_{t-2, t-1}+\beta_{t-1, t-1} \Big )\nonumber\\
 -i   \sin(\theta) \Big (& \alpha_{-t-1, t-1}+\alpha_{-t, t-1} + \cdots \alpha_{t-2, t-1}+\alpha_{t-1, t-1} \Big ) \Big ] | 1\rangle, 
\end{align}
and reorganising the terms on the right hand side, it can be written as   
\begin{align}
\label{eq:14}
  |\Psi^f_p (t )\rangle& = \sum_{j= -(t-1)}^{t-1}  
      ( \cos(\theta) \alpha_{j, t-1}  -i\sin(\theta) \beta_{j, t-1}  ) |0\rangle  \nonumber \\
     \quad +\Big[\cos(\theta) &(\alpha_{t, t-1}+\alpha_{t+1, t-1})  -i\sin(\theta)(\beta_{t, t-1}+\beta_{t+1, t-1})\Big]|0\rangle\nonumber\\
   &     +\sum_{j= -(t-1)}^{t-1}  ( \cos(\theta) \beta_{j, t-1}  -i\sin(\theta) \alpha_{j, t-1}   ) |1\rangle  \nonumber \\
              \quad&+\Big [ \cos(\theta) (\beta_{-t-1, t-1} + \beta_{-t, t-1}) \nonumber \\
              &-i \sin(\theta) (\alpha_{-t-1, t-1} + \alpha_{-t, t-1}) \Big]|1\rangle.
        \end{align}
At time $t-1$, the amplitudes $\alpha$ and $\beta$ at position $j < -(t-1)$ and $j > (t-1)$  will be zero, which reduces the preceding expression to 
\begin{align}
\label{eq:15}
  |\Psi^f_p (t)\rangle   &     =  \sum_{j= -(t-1)}^{t-1}  \Big[  \cos(\theta) 
  \left  ( \alpha_{j, t-1} |0\rangle +  \beta_{j, t-1}  |1\rangle \right ) \nonumber \\
& \qquad\qquad\qquad-i \sin(\theta)  \left (  \beta_{j, t-1} |0\rangle +  \alpha_{j, t-1}  |1\rangle \right ) \Big ]   \nonumber  \\
     &     =  e^{-i \theta \cdot \sigma_x} |\Psi^f_p (t-1)\rangle.
                                         \end{align}
Iterating the same process of expressing the amplitudes at position $j$ through the contributing amplitudes from the neighbouring positions $j\pm1$ at the previous time a  further $t-1$ times then leads to
\begin{align}
  |\Psi^f_p (t)\rangle &
           =  e^{-i t \theta \cdot \sigma_x} |\Psi^{in}_p\rangle.
                                              \end{align}
This validates the statement of the Theorem.
\hfill\qed
\end{proof}

By considering the Euler expansion of the last expression
\begin{equation}
   |\Psi^f_p (t)\rangle=\Big(\cos(t\theta) \cdot {\mathbbm 1}  -i \sin(t\theta) \cdot \sigma_x \Big )  |\Psi^{in}_p\rangle,
\end{equation}
a few special cases following from Theorem\,\ref{Th1} can be immediately understood 
\begin{enumerate}
  \item   if $\alpha_0=\pm\beta_0$,\qquad $|\Psi_p^f(t )\rangle = e^{-i\theta t} |\Psi_{p}^{in}\rangle$  \qquad $\forall~t$ 
  \item  if $\alpha_0 \neq \pm \beta_0$ 
   \begin{align}
\qquad |\Psi_p^f(t)\rangle &= (-1)^n |\Psi_{p}^{in}\rangle  &&\text{if}~~~ t  = \pi n/\theta \nonumber \\
\qquad  |\Psi_p^f(t)\rangle &= (-1)^ni \sigma_x |\Psi_{p}^{in}\rangle   &&\text{if}\quad t  = \pi (2n+1)/2\theta,\nonumber 
\end{align}
where $n=\{0, 1, 2, 3, \cdots \}$. 
\end{enumerate}

To better understand the full process of the quantum walk as quantum memory we show in Fig.~\ref{fig:2} the numerically obtained probabilities for finding the basis state $|0\rangle$ as function of $\eta$ and $\delta$ and for a fixed coin value, $\theta=\pi/6$.  The probability of finding state $|1\rangle$ can simply be inferred by realizing that the probabilities are normalized to one. In Fig.~\ref{fig:2}(a)  we have set $\eta=0$ and the thick light green curve shows the probability at $t=0$ (initial state). As expected, the final probabilities for $t=\pi n/\theta=6n$ are identical with the initial ones     (dashed blue curve), whereas for times $t=\pi(2n+1)/2\theta=6n+3$ the probability of finding the state $|0\rangle$ is equal to the probability of finding $|1\rangle$ initially (red curve). For times between these two (shown are $t=6n+1$ and $t=6n+2$), the final state has no easy relation to the initial one, except at $\delta=\pi/4$ and $\delta=3\pi/4$. 

Fig.\,\ref{fig:2}(b) shows the probability of finding state $|0\rangle$ as a function of $\eta$ for two different fixed values of $\delta$. In the initial state this quantity is independent of $\eta$ and this is indicated by the two  thick horizontal lines, light green for $\delta=\pi/6$ and light orange for $\delta=\pi/4$. For both cases one can see that again for $t=\pi n/\theta=6n$ the initial state is perfectly recovered, for $t=\pi(2n+1)/2\theta=6n+3$ the probability of finding $|0\rangle$ is equal to the one of finding $|1\rangle$ initially and for other times no easy relation exists (examples for $\delta=\pi/6$ and $t=6n+2$ and $t=6n+5$ are shown). All these results are consistent with Theorem\,\ref{Th1} and we have found this to be true independently of the choice of coin operation. 

The above model for using a quantum walk as a quantum memory has an additional aspect in that it helps to keep the information stored in the qubit secure against eavesdropper who does not have access to the complete position space. This kind of security is different from the standard encryption models, which nevertheless could also be added on top of the above scheme. However, it has also the inherent limitation of only allowing perfect retrieval at certain times $t$, which depend on the coin parameter $\theta$. Additionally, the size of the position space required to store the information increases linearly with time, making it  experimentally challenging to store it for longer durations. In the following we will show how to overcome both these shortcomings.

\section{ Localised quantum walk as quantum memory}
\label{LQWQM}

Perfect recovery at any time $t$ independently of $\theta$ can be achieved by encoding and decoding the initial and the final state by a Hadamard operation, and introducing temporal disorder to the walk evolution will allow to compactly store the information even for longer times. The latter can be achieved by using a time dependent coin operation $B(\theta_t)$ (temporal disorder) with a different value of $\theta_t$ (uniformly distributed) in the range $-\pi/2 \leq \theta_t \leq \pi/2 $ and randomly chosen for each time, $t$\,\cite{JM10, Cha11a, Cha12+}, 
\be
|\Psi_t \rangle = W(\theta_{t})\dots W(\theta_{3})W(\theta_{2})W(\theta_{1})|\Psi_{in}\rangle.
\ee
The collection process can be described using the same operator $W_T$  as  given in Eq.~\eqref{transfer}. In Fig.\,\ref{fig:3} we show the schematic for using this localized quantum walk with temporal disorder  as quantum memory and although it might seem surprising initially that the disorder does not affect the phases irrevocably, the following theorem shows the possibility of full retrieval of the initial state.
\begin{figure}[tb]
  \includegraphics[width=8.5cm]{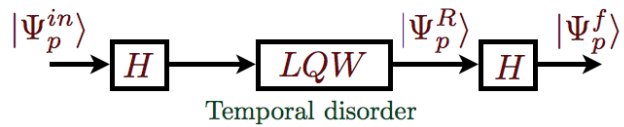} 
 \caption{\footnotesize{{\bf Schematic for the localized quantum walk (LQW) as quantum memory.} 
  With this protocol any arbitrary initial state of the qubit can be securely stored and retrieved at any desired time. \label{fig:3}}}
\end{figure}


\begin{theo}
\label{Th2} 

A discrete-time quantum walk using temporal disorder during the evolution can compactly store any arbitrary initial state $|\Psi^{in}_p\rangle= \alpha_0 | 0 \rangle + \beta_0 |1 \rangle$ encoded using a Hadamard operation $H =\frac{1}{\sqrt 2}\begin{bmatrix}  
  1      &      ~~ 1
  \\ 1 &     -1 
\end{bmatrix}$. At any time $t$ the retrieved state can be decoded using $H$ to recover the initial state,
\bea
|\Psi_p^{f} (t  )\rangle = \begin{bmatrix}  
  ~e^{-i \Theta }     &       0
  \\ 0 &     ~e^{i \Theta} 
\end{bmatrix} |\Psi_p^{in}\rangle,
\eea
where $\Theta = \sum_{k=1}^{t} \theta_k$ and $-\frac{\pi}{2} \leq \theta_k \leq \frac{\pi}{2}$.

\end{theo}

\begin{figure}[tb]
\bc 
\includegraphics[width=8.8cm]{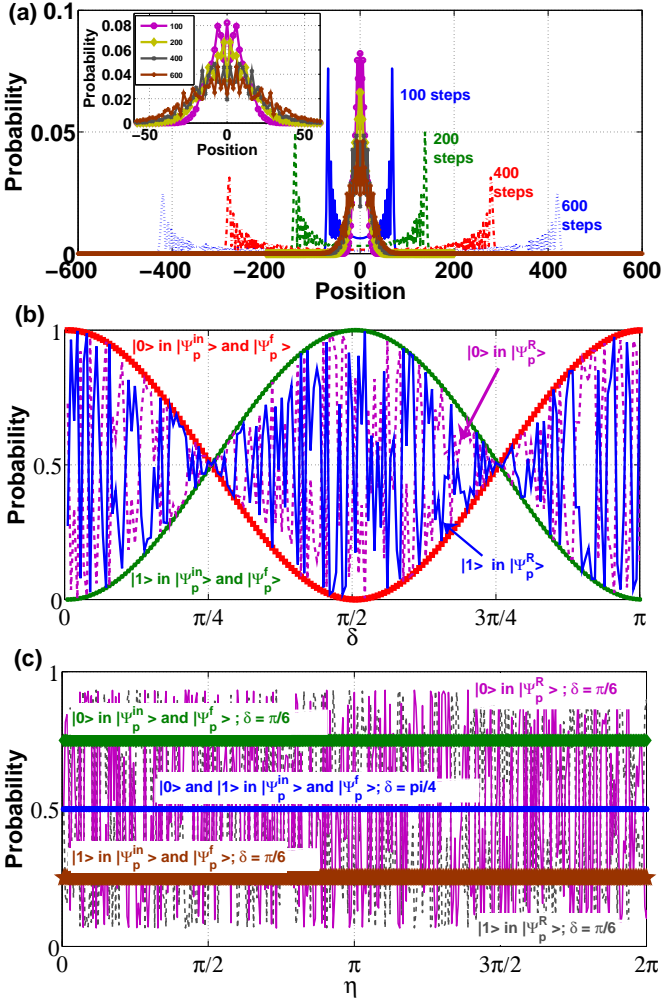} 
\ec
\vskip -0.7cm
\caption{\footnotesize{{\bf Localized quantum walk in position space and the probabilities of the basis states in the initial and the final state.} 
(a) Spatial distribution of the qubit with temporal disordered evolution (full lines) vs.~ordered evolution (broken lines). The inset shows a closeup of the localized part. (b) and (c)  Probability of finding the basis states $|0\rangle$ and $|1\rangle$ in the retrieved state $|\Psi_p^{R}\rangle$ (fast oscillating curves) and the decoded state $|\Psi_p^{f}\rangle$ (slowly oscillating curves) for any value of $t$ and as function of  $\delta$ when $\eta =0$ in (b) and as function of $\eta$ for different value of $\delta$ in (c). After decoding, a perfect retrieval of the initial probabilities is obtained}.\label{fig:4}}
\end{figure}

 \begin{proof} 
We will show the validity of Theorem\,\ref{Th2} in this proof. The complete evolution process is described as,
\begin{align}
  & |\Psi^f_p(t)\rangle \otimes |R\rangle = \Big [ H \otimes  {\mathbbm 1} \Big ]|\Psi^{R}_p(t)\rangle \otimes |R\rangle   
     \nonumber \\
  &=\Big [ H \otimes  {\mathbbm 1} \Big ]W_T \Big [W(\theta_t)   \cdots W(\theta_1)\Big ] |\Psi^{in}_p\rangle_E
 \otimes | j = 0 \rangle,
\end{align}
where $|\Psi^{R}_p(t)\rangle$ is the state retrieved at vertex $R$ after time $t$ before decoding and 
\begin{align}
|\Psi_p^{in} \rangle_E =  H  |\Psi_p^{in} \rangle = \frac{(\alpha_0 + \beta_0)}{\sqrt 2}|0\rangle + \frac{(\alpha_0 - \beta_0)}{\sqrt 2}|1\rangle 
\end{align}
is the encoded initial state. 
In the same way as detailed in the proof for Theorem\,\ref{Th1}, we can write the retrieved state at vertex $R$ as, 
\begin{align}
\label{eq19}
|\Psi^{R}_p(t )\rangle =& \sum_{j= -t}^{t}  \Big[ \alpha_{j, t} |0\rangle   +  \beta_{j, t} |1\rangle \Big ],
\end{align}
and with the backward iterative approach we get
\begin{align}
\label{Eq1Th2}
  |\Psi^R_p &(t)\rangle  = \nonumber \\
  &  \;\;        \sum_{j= -t}^{t} \Bigg[\Big ( \cos(\theta_t) \alpha_{j+1, t-1} -i\sin(\theta_t) \beta_{j+1, t-1}\Big ) |0\rangle  \nonumber \\
        & \;\;  + \Big ( \cos(\theta_t) \beta_{j-1, t-1} -i\sin(\theta_t) \alpha_{j-1, t-1} \Big ) |1\rangle \Bigg ].
\end{align}
Following the same procedure as in  Eqs.\,(\ref{eq:13}), \,(\ref{eq:14}) and \,(\ref{eq:15}) in Theorem\,\ref{Th1} leads to,
\begin{align}
|\Psi^{R}_p(t)\rangle 
&=   e^{-i \theta_t \cdot \sigma_x} |\Psi^{R}_p(t-1)\rangle,
\end{align}
and after iterating this for a further $t-1$ times we obtain
\begin{align}
|\Psi^{R}_p(t)\rangle 
&=   e^{-i \theta_t \cdot \sigma_x} |\Psi^{R}_p(t-1)\rangle = e^{-i \Theta \cdot \sigma_x}  |\Psi_p^{in} \rangle_E,
\end{align}
where $\Theta = \sum_{k=1}^{t} \theta_k$.
This can be simplified as
\begin{align}
|\Psi^{R}_p(t)\rangle  =&  \frac{\cos(\Theta)}{\sqrt 2}  \Big[ \alpha_0 (|0\rangle + |1\rangle ) + \beta_0 (|0\rangle  - |1\rangle) \Big] \nonumber \\
 - &\frac{i\sin(\Theta)}{\sqrt 2}  \Big[ \alpha_0 ( |0\rangle + |1\rangle ) - \beta_0 (|0\rangle  -  |1\rangle) \Big] \nonumber  \\
= e^{-i \Theta} & \alpha_0 \Bigg (\frac{|0\rangle + |1\rangle}{\sqrt 2} \Bigg ) + e^{i\Theta} \beta_0 \Bigg (\frac{|0\rangle - |1\rangle}{\sqrt 2} \Bigg ),
                                          \end{align}
and the final state after decoding using $H$ is
\begin{align}
|\Psi^{f}_p(t)\rangle = H  |\Psi^{R}_p(t)\rangle =   \begin{bmatrix}  
  ~e^{-i \Theta }     &       0
  \\ 0 &     ~e^{i \Theta} 
\end{bmatrix} |\Psi^{in}_p\rangle.
\end{align}
This validates the statement of the theorem $\forall ~t$.
\hfill\qed
\end{proof}

In Fig.\,\ref{fig:4}(a), we  show numerical results for the probability distribution of the standard quantum walk evolution ($\theta = \pi/4$, broken lines) and the walk with temporal disordered evolution for different number of steps (full lines). Localization of the probability distribution around the initial position is seen for disordered evolution and the inset shows that the widths of the localized distributions do not significantly increase for larger numbers of steps. The probabilities of finding the basis states $|0\rangle$ and $|1\rangle$ of the initial state in the final state after retrieval ($|\Psi_p^{R}\rangle$) and after decoding ($|\Psi_p^{f}\rangle$) are shown in Figs.\,\ref{fig:4}(b) and (c) as function of $\delta$ (for $\eta =0$) and $\eta$, respectively, for all $t$. For the retrieved state a strongly fluctuating distribution is seen, which nevertheless allows perfect retrieval of initial state after decoding. These numerical observations are consistent with Theorem \ref{Th2} and show the potential of localized quantum walks as a secured quantum memory.

The existence of symmetries in the quantum walk is the reason for the possibility of obtaining the global phase factor, $\Theta$, in Theorem\,\ref{Th2}. This can be seen in the proof where during each backward iteration step \,(Eq.\,(\ref{Eq1Th2})) a symmetry in the contribution from neighbouring positions to position $j$ is preserved. Because of this symmetry we obtain $\Theta = \sum \theta_k \approx  0$. This symmetry, however, would be broken in the presence of spatial disorder, in which case the initial state cannot be retrieved using any known techniques.  Therefore, the use of temporal disorder is fundamentally important in the scheme presented above.

Let us finally comment on the possibilities of a physical realization of scheme presented above. As the crucial requirement is to be able to implement coin operations which can have any value for $\theta$, the recently realised quantum walks using an optically trapped cold atom as the qubit\,\cite{K09} are one system that holds this possibility: temporally disordered coin operations can be realized by replacing the $\pi/2$ pulses used in the standard evolution by pulses of random length varying between $-\pi/2$ and $\pi/2$ at each step. This will lead to localisation of the atom in the optical lattice around the initial position and therefore securely store the qubit. The information can then be retrieved at any desired time by turning off the optical potential and letting the contributions from different lattices sites interfere. However, in principle this physically demonstrates the proposed protocol for quantum memory, the spatial nodes $|j\rangle$ stay too close to each other compromising the security of the scheme. Advancement in experimental techniques and carefully designed protocols in other physical systems could lead to, a wide separation of the spatial modes implementing a secured quantum memory practically effective.

\section{Conclusions} 
\label{conclusion}
In conclusion, we have suggested the use of quantum walks as quantum memory and shown that the initial state of any unknown qubit can be retrieved after a quantum walk evolution. However, if the evolution is a standard one, the retrieval time $t$ is a periodic function of the coin value $\theta$, and the size of the position space to store the information increases linearly with time, making it an experimentally difficult task to keep the information for longer times. 

To address both these shortcomings, we have shown that by using a Hadamard operation to encode and decode the initial state and the retrieved state, respectively, the stored information can be read out perfectly at any stage of the discrete walk.  To curb the linear growth of the position space for longer storage times, we have suggested to localize the quantum walk using temporal disorder and shown that perfect recovery is still possible.  Since the localized distribution will still have a finite width in position space, this maintains the inherent security of the protocol against a simple eavesdropping attack. 

Finally, let us point out that our scheme is trivially consistent with the requirement of a quantum memory in a quantum computer to have a structure that  facilitates controlled quantum evolutions\,\cite{PK11}.  However, our memory's abilities for information theoretic cooling and fault-tolerant operation are still to be explored in future work. As our scheme can be realized using random coin values  $\theta$ at each step of the temporal evolution (as shown in Theorem\,\ref{Th2}),  precise control over the coin parameter is not required. This significantly lowers the barrier for experimental implementation.


\end{document}